\documentclass[manuscript,nonacm]{acmart}

\usepackage{amsmath,mathtools,bbm}
\usepackage{thm-restate} % for restating theorems in the appendix
\newtheorem{remark}{Remark}

\setcopyright{none}
\copyrightyear{2021}
\acmYear{2021}
\acmDOI{}
\acmConference[]{}{}{}
\acmBooktitle{}
\acmPrice{}
\acmISBN{}
\newcommand{\setN}{\mathbb{N}}
\newcommand{\setZ}{\mathbb{Z}}

\newcommand{\eqby}[1]{\stackrel{\textrm{{\normalfont\tiny{#1}}}}{=}}
\newcommand{\eqdef}{\eqby{def}}

\newcommand{\norm}[1]{{\mathopen{\|}#1\mathclose{\|}}}

\begin{document}

\title{State Complexity of Protocols With Leaders} %TODO Please add

%\titlerunning{}%optional, please use if title is longer than one line

\author{Jérôme Leroux}
\email{jerome.leroux@labri.fr}
\affiliation{%
  \institution{LaBRI, CNRS, Univ. Bordeaux}
  \city{Talence}         % ACM template requires city and
  \country{France}   % country for affiliation
}

%\thanks{The author is supported by the grant ANR-17-CE40-0028 of the French National Research Agency ANR (project BRAVAS)}

\begin{abstract}
  Population protocols are a model of computation in which an arbitrary number of anonymous finite-memory agents are interacting in order to decide by stable consensus a predicate. In this paper, we focus on the counting predicates that asks, given an initial configuration, whether the number of agents in some initial state $i$ is at least $n$. In 2018, Blondin, Esparza, and Jaax shown that with a fix number of leaders and interaction-width, there exists infinitely many $n$ for which the counting predicate is stably computable by a protocol with at most $O(\log\log(n))$ states. We provide in this paper a matching lower-bound (up to a square root) that improves the inverse-Ackermannian lower-bound presented at PODC in 2021. 
\end{abstract}

\maketitle

\newcommand{\used}[1]{\operatorname{used}(#1)}
\newcommand{\vr}[1]{#1}

\section{Introduction}
Population protocols were introduced by Angluin, Aspnes, Diamadi, Fischer, and Peralta in \cite{DBLP:conf/podc/AngluinADFP04,DBLP:journals/dc/AngluinADFP06} to study the computational power of networks of resource-limited mobile agents. In this model, each agent has a state in a finite set of states. When agents interact, their states are updated accordingly to a finite interaction table. This table corresponds intuitively to a conservative Petri net (the Petri net is conservative since the number of agents is preserved by each transition) where each line of the interaction table is matched by a transition of the Petri net. In this model, an agent may accept or reject depending only on it own state. A population protocol is said to be stably computing a predicate, if for any initial configurations, eventually and forever, under some natural fairness conditions, either all agents accept or all agents reject. Moreover, this outcome should only depend on the initial configuration and not on the way interactions are performed.

Deciding if a protocol stably computes some unknown predicate is a problem called the \emph{well-specification problem}. This problem was proved to be decidable in~\cite{DBLP:conf/concur/EsparzaGLM15,DBLP:journals/acta/EsparzaGLM17} by observing that well-specification problem is equivalent to the reachability problem for Petri nets up to elementary reductions. Since this last problem was recently proved to be Ackermannian-complete~\cite{DBLP:journals/corr/abs-2104-12695,DBLP:journals/corr/abs-2104-13866}, it means that deciding the well-specification problem is Ackermannian-complete. Intuitively, population protocols maybe intrinsically very complicated.

Despite this Ackermannian complexity result, in~\cite{DBLP:journals/dc/AngluinAER07}, Angluin, Aspnes, Eisenstat, and Ruppert have shown that predicates stably computable by population protocols cannot be more complicated than the one definable in the Presburger arithmetic. Combined with \cite{DBLP:conf/podc/AngluinADFP04,DBLP:journals/dc/AngluinADFP06}, it follows that predicates stably computable by population protocols are exactly the predicates definable in the Presburger arithmetic.

Since deciding if a population protocol is stably computing some Presburger predicate is Ackermannian-complete, a natural question is the conciseness of population protocols. Intuitively, is it possible to define a population protocol computing predicates that are very complex compared to the number of states of the protocol ? This problem is related to the so-called \emph{state complexity} of a Presburger predicate intuitively defined as the minimal number of states of a population protocols deciding it.

State complexity upper-bounds are obtained thanks to algorithms computing from predicates protocols stably computing it with a number of states as small as possible. In~\cite{DBLP:conf/stacs/BlondinEGHJ20}, by revisiting the construction of population protocols deciding Presburger predicates, some improvement on state complexity upper-bounds was derived. On the other side, state complexity lower-bounds is also a difficult task since such a bound requires to prove that there is no way to stably compute a predicate with a given amount of states. In this context, focusing on the state complexity of simple Presburger predicates is a natural question. The simplest non trivial Presburger predicates are clearly the counting predicates that corresponds to the set of configurations such that the number of agents in a given state is larger than or equal to some positive number $n$. In 2018, Blondin, Esparza, and Jaax shown in~\cite{DBLP:conf/stacs/BlondinEJ18} that with a fix number of leaders and interaction-width (the number of agents that can interact at each interaction step), there exists infinitely many $n$ for which the counting predicate is stably computable by a protocol with at most $O(\log\log(n))$ states.

This state complexity upper-bound was recently completed by a state complexity lower-bound in~\cite{DBLP:conf/podc/CzernerE21}. In that paper, Czerner and Esparza shown that the number of states of a population protocol deciding a counting predicate with a bounded number of leaders and a bounded interaction-width is at least $\Omega(A^{-1}(n))$ where $A$ is some Ackermannian function, leaving a gap between the $O(\log\log(n))$ upper-bound and the $\Omega(A^{-1}(n))$ lower-bound.

\medskip

\textbf{Main result}.
In this paper, we follow the model of protocols introduced by Dana Angluin, James Aspnes, and David Eisenstat in 2006 that allows agents creations and destructions~\cite{DBLP:conf/podc/AngluinAE06}. We slightly extends that model to allow leaders. We close the previously mentioned state complexity gap by proving that for any $h<\frac{1}{2}$, and under a fix number of leaders and a bounded number of interaction-width, any protocol stably computing a counting predicate requires at least $\Omega((\log\log(n))^h)$ states.

\medskip

\textbf{Outline}.
In Section~\ref{sec:protocol} we recall some basic definitions and results about protocols. In Section~\ref{sec:interaction}, we introduce the notion of communication-width of a protocol and show that protocol with finite communication-width are naturally related to the model of Petri nets. In Section~\ref{sec:complexity} we introduce the state complexity problem and show that counting the number of states of protocols without taking into account the number of leaders or the communication-width is not relevant. In Section~\ref{sec:stab} we introduce the notions of stabilized configurations and show that those configurations are characterized by their small values. Those results are obtained thanks to Rackoff's techniques originally introduced in the context of the coverability problem for Petri nets. Section~\ref{sec:bot} contains the central technical lemma. It is a lemma about Petri nets that intuitively shows that from any initial configuration we can reach with short executions kind of bottom configurations. Section~\ref{sec:PNS} recalls the model of Petri net with control-states and provides a result on small cycles satisfying some properties. This last result is obtained by introducing a linear system and by applying Pottier's techniques~\cite{Pottier:1991:MSL:647192.720494} in order to obtain small solutions for that linear system. Results of the previous sections are combined in Section~\ref{sec:main} to obtain our state complexity lower-bound. Some related open problems and future work are presented in Section~\ref{sec:conc}.

\section{Protocols}\label{sec:protocol}
In this section, we introduce our model of protocols that slightly differs from the one introduced in~\cite{DBLP:conf/podc/AngluinAE06} by allowing leaders. We extend the definition of stably computing a predicate with such a protocol. This definition is a straight-forward extension of the one given in~\cite{DBLP:conf/podc/AngluinAE06} that generalizes several definitions of \emph{fair computations} in particular in the context of unconservative protocols.

\medskip

\newcommand{\dom}[1]{\operatorname{dom}(#1)}

Let $P$ be a finite set of elements called \emph{states}. A \emph{$P$-configuration} (or just a \emph{configuration} if $P$ is clear from the context) is a mapping in $\setN^P$. Given a configuration $\rho$, the number $|\rho|\eqdef\sum_{p\in P}\rho(p)$ is called the \emph{number of agents} in $\rho$. Let $Q$ be another finite set of states. We associate with a $P$-configuration $\rho$ the $Q$-configuration $\rho|_Q$ defined for every $q\in Q$ by $\rho|_Q(q)=\rho(q)$ if $q\in P$ and zero otherwise. Notice that $Q$ is not necessarily a subset of $P$. Given $p\in P$ we simply denote by $p|_P$ (or just $p$ when $P$ is clear from the context) the mapping in $\setN^P$ that maps $p$ on $1$ and the other states on zero. The sum $\alpha+\beta$ of two configurations $\alpha,\beta$ and the product $n.\rho$ where $n\in \setN$ and $\rho$ is a configuration are defined component-wise as expected.

\medskip

Let $R$ be a binary relation on $P$-configurations for some finite set of states $P$. We say that $R$ is \emph{additive} if $(\alpha,\beta)\in R$ implies $(\alpha+\rho,\beta+\rho)\in R$ for every configuration $\alpha,\rho,\beta$. As usual, $R$ is called a \emph{preorder} if it is \emph{reflexive} and \emph{transitive}. We also say that $R$ is \emph{conservative} if $|\alpha|=|\beta|$ for every $(\alpha,\beta)\in R$.

\medskip

A \emph{protocol} is a tuple $(P,\xrightarrow{*}, \rho_L,I,\gamma)$ where $P$ is a finite set of \emph{states}, $\xrightarrow{*}$ is an additive preorder on the $P$-configurations,  $\rho_L$ is a $P$-configuration called the \emph{configuration of leaders}, $I\subseteq P$ is the set of \emph{initial states}, and $\gamma:P\rightarrow \{0,\star,1\}$ is the \emph{output function}. The value $|\rho_L|$ is called the number of \emph{leaders}. A protocol is said to be \emph{leaderless} when this number is zero. A configuration of the form $\rho_L+\rho|_P$ with $\rho\in\setN^I$ is called an \emph{initial configuration}. A protocol is said to be \emph{conservative} if $\xrightarrow{*}$ is conservative. The function $\gamma$ is extended on any $\rho\in \setN^P$ by:
$$\gamma(\rho)=\{j\in \{0,\star,1\} \mid \exists p\in P ~\rho(p)>0~\wedge~\gamma(p)=j \}$$

\medskip

For $j\in\{0,1\}$, we introduce the following sets $S_j$ called the $j$-output stable configurations. Notice that $S_0$ and $S_1$ are not defined exactly the same way in order to manage the zero configuration. With our definition we interpret the output of the zero configuration as $0$. Notice that the definition of $0$-output stable configurations introduced in~\cite{DBLP:conf/podc/AngluinAE06} does not care about the zero configuration since protocols are semantically restricted to non zero configurations. We do not introduce a set $S_\star$ since intuitively the element $\star$ in the image of $\gamma$ is interpreted as an undetermined output.  
\begin{align*}
  S_0&\eqdef\{\alpha\in \setN^P \mid \forall \beta~ \alpha\xrightarrow{*}\beta \Rightarrow \gamma(\beta)\subseteq \{0\}\}\\
  S_1&\eqdef\{\alpha\in \setN^P \mid \forall \beta~ \alpha\xrightarrow{*}\beta \Rightarrow \gamma(\beta)= \{1\}\}
\end{align*}
  
\medskip

A \emph{predicate} is a mapping $\phi:\setN^I\rightarrow\{0,1\}$. We say that a protocol stably computes $\phi$ if for every $\rho\in\setN^I$ and for every $\alpha\in\setN^P$ such that $\rho_L+\rho|_P\xrightarrow{*} \alpha$, there exists $\beta\in S_{\phi(\rho)}$ such that $\alpha\xrightarrow{*}\beta$. A predicate $\phi$ is \emph{stably computable} if there exists a protocol that stably computes it.

\begin{remark}
  In~\cite{DBLP:conf/podc/AngluinAE06}, predicates that are stably computable by leaderless protocols restricted to functions $\gamma$ such that $\gamma(P)\subseteq \{0,1\}$ are proved to be exactly the predicates definable in the Presburger arithmetic. We think that this property can be extended to any protocol with leaders and using the $\star$ element in the definition of $\gamma$. We left this problem open since it is outside of the scope of this paper.
\end{remark}

\section{Interaction-Width}\label{sec:interaction}
The class of additive preorders is central for defining protocols as shown in the previous section. Additive preorders are a natural generalization of the Petri net reachability relations used in classical population protocols. For parameterized complexity purposes, we introduce in this section the notion of interaction-width that intuitively limits the number of agents that can communicate together in a single interaction step. As expected additive preorders with finite interaction-width are exactly the reachability relations of Petri nets.

\medskip

A \emph{$P$-transition} (or simply a \emph{transition} when the finite set $P$ is clear from the context) is a \emph{pair} $t\eqdef (\alpha_t,\beta_t)$ of $P$-configurations. Given such a transition, we introduce $|t|\eqdef \max\{|\alpha_t|,|\beta_t|\}$ called the \emph{interaction-width} of $t$. We associate with a transition $t$ the binary relation $\xrightarrow{t}$ over the configurations defined by $\alpha\xrightarrow{t}\beta$ if there exists a configuration $\rho$ such that $\alpha=\alpha_t+\rho$ and $\beta=\beta_t+\rho$. Notice that $\xrightarrow{t}$ is the minimal for the inclusion additive binary relation that contains $t$. Given a word $\sigma=t_1\ldots t_k$ of transitions, we introduce the binary relation $\xrightarrow{\sigma}$ over the configurations defined by $\alpha\xrightarrow{\sigma}\beta$ if there exists a sequence $\rho_0,\ldots,\rho_k$ of configurations such that:
$$\alpha=\rho_0\xrightarrow{t_1}\rho_1\cdots \xrightarrow{t_k}\rho_k=\beta$$

\medskip

\newcommand{\com}[1]{\operatorname{width}(#1)}

An additive preorder $R$ is said to have a \emph{finite interaction-width} if there exists $m\in\setN$ such that for every $(\alpha,\beta)\in R$, there exists a word $\sigma$ of transitions in $R$ with an interaction-width bounded by $m$ and such that $\alpha\xrightarrow{\sigma}\beta$. The minimal $m$ satisfying this property is called the \emph{interaction-width} of $R$, and it is denoted as $\com{R}$. When $R$ does not have a finite interaction-width, we define $\com{R}\eqdef\omega$. The \emph{interaction-width of a protocol} is defined as the interaction-width of its implicit additive preorder.

\medskip

Finite interaction-width additive preorders are related to Petri nets as follows. A $P$-Petri net $T$ (or simply a Petri net when $P$ is clear from the context) is a finite set of $P$-transitions. The reachability relation of a Petri net $T$ is the binary relation $\xrightarrow{T^*}$ over the configurations defined by $\alpha\xrightarrow{T^*}\beta$ if there exists $\sigma\in T^*$ such that $\alpha\xrightarrow{\sigma}\beta$. Observe that the reachability relation of a Petri net $T$ is an additive preorder with an interaction-width bounded by $\max_{t\in T}|t|$. Moreover, if $R$ is an additive preorder with a finite interaction-width then $R$ is the reachability relation of the Petri net $\{t\in R \mid |t|\leq \com{R}\}$. It follows that the class of additive preorders with finite interaction-widths is equal to the class of Petri net reachability relations.

\section{State-complexity of Protocols}\label{sec:complexity}
The state-complexity of a predicate is intuitively the minimal number of states of a protocol that stably computes it. In this section, we show that the state-complexity must involve the interaction-width and the number of leaders to discard trivial results.

\medskip

Since our paper focuses on the so-called counting predicates, let us first introduce those predicates. The \emph{$(i\geq n)$ predicate} where $i$ is a state and $n$ a positive natural number is the predicate $\phi_{i\geq n}:\setN^I\rightarrow\{0,1\}$ where $I\eqdef\{i\}$ satisfying $\phi_{i\geq n}(\rho)=1$ if, and only if $\rho(i)\geq n$ for every configuration $\rho\in\setN^I$. Such a predicate is called a \emph{counting predicate}.

\medskip

The following two examples show that counting the minimal number of states of protocols stably computing a counting predicate without taking into account the interaction-width or the number of leaders is not relevant.
\begin{example}
  This example shows that the predicate $\phi_{i\geq n}$ is stably computable by a leaderless conservative protocol with a number of states bounded $2$. We introduce the protocol $(P,\xrightarrow{*},\rho_L,I,\gamma)$ where $P\eqdef\{i,p\}$ for some state $p\not=i$, $\rho_L$ is the zero configuration, $I\eqdef\{i\}$, $\gamma^{-1}(\{0\})=\{i\}$, $\gamma^{-1}(\{1\})=\{p\}$, and $\xrightarrow{*}$ is the additive preorder defined by $\alpha\xrightarrow{*}\beta$ if there exists $m\in\setN$ such that $\beta+m.i=\alpha+m.p$ and such that $m=0\vee |\alpha|\geq n$. In fact, notice that for every $\rho\in \setN^I$ and $\alpha\in\setN^P$ such that $\rho|_P\xrightarrow{*}\alpha$, then either $\rho(i)<n$ and in that case $\alpha=\rho|_P$ is a $0$-output stable configuration, or $\rho(i)\geq n$ and in that case $\alpha\xrightarrow{*}\beta$ where $\beta$ is the $1$-output stable configuration defined as $\rho(i).p$. We have proved that the protocol stably computes $\phi_{i\geq n}$. The interaction-width of the previous protocol is bounded by $n$. In fact, just observe that the reachability relation of the Petri net $\{(\rho+i,\rho+p) \mid \rho\in\setN^P \wedge |\rho|=n-1\}$ is equal to $\xrightarrow{*}$. It follows that $\com{\xrightarrow{*}}\leq n$. In fact, one can easily prove that $\com{\xrightarrow{*}}=n$.
\end{example}

\begin{example}
  This example shows that the predicate $\phi_{i\geq n}$ is stably computable by a conservative protocol with a number of states bounded $6$ and an interaction-width bounded by $2$. We introduce the protocol $(P,\xrightarrow{T^*},n.\bar{i},I,\gamma)$ where $P\eqdef\{i,\bar{i},p,\bar{p},q,\bar{q}\}$, $I\eqdef\{i\}$, $\gamma^{-1}(\{1\})=\{i,p,q\}$, $\gamma^{-1}(\{0\})=\{\bar{i},\bar{p},\bar{q}\}$, and $T\eqdef\{t,t_p,\bar{t}_p,t_q,\bar{t}_q,t_{\bar{p}},t_{\bar{q}}\}$ is the Petri net defined as follows:
  $$
  t\eqdef(i+\bar{i},p+q)
  ~~~
  \begin{array}{r@{}l@{}}
    t_p&\eqdef(\bar{p}+i,p+i)\\
    \bar{t}_p&\eqdef(p+\bar{i},\bar{p}+\bar{i})
  \end{array}
  ~~~
  \begin{array}{r@{}l@{}}
    t_q&\eqdef(\bar{q}+i,q+i)\\
    \bar{t}_q&\eqdef(q+\bar{i},\bar{q}+\bar{i})
  \end{array}
  ~~~
  \begin{array}{r@{}l@{}}
    t_{\bar{q}}&\eqdef(p+\bar{q},p+q)\\
    t_{\bar{p}}&\eqdef(q+\bar{p},q+p)
  \end{array}
  $$
  Notice that each transition, except $t$, can only change the bar statue of an agent in states $p$ and $q$. Let $\rho\in\setN^I$ and let $\alpha$ be a configuration such that $n.\bar{i}+\rho|_P\xrightarrow{T^*}\alpha$ and let us prove that there exists a $\phi_{i\geq n}(\rho)$-output stable configuration $\beta$ such that $\alpha\xrightarrow{T^*}\beta$. First of all, by executing $\min\{\alpha(i),\alpha(\bar{i})\}$ times the transition $t$ from $\alpha$, we can assume without loss of generality that $\alpha(i)=0$ or $\alpha(\bar{i})=0$. Observe that $\alpha(\bar{i})=0$ if and only if $\phi_{i\geq n}(\rho)=1$. Assume first that $\alpha(\bar{i})>0$. In that case, by executing the transitions $\bar{t}_p$ and $\bar{t}_q$ the right number of times, we get from $\alpha$ a configuration $\beta$ such that $\beta(i)=\beta(p)=\beta(q)=0$. Notice that $\beta$ is $0$-output stable. Next, assume that $\alpha(i)>0$. In that case, by executing the transitions $t_p$ and $t_q$ the right number of times, we get a configuration $\beta$ such that $\beta(\bar{i})=\beta(\bar{p})=\beta(\bar{q})=0$. This configuration is $1$-output stable. Finally, assume that $\alpha(i)=0$ and $\alpha(\bar{i})=0$. Since $n>0$, by identifying the last time transition $t$ is executed, we deduce that there exist configurations $\mu,\delta$ such that $n.\bar{i}+\rho|_P\xrightarrow{T^*}\mu\xrightarrow{t}\delta\xrightarrow{T_0^*}\alpha$ such that $\delta(i)=\delta(\bar{i})=0$ where $T_0\eqdef \{t_{\bar{q}},t_{\bar{p}}\}$. Notice that $\delta(p)>0$ and $\delta(q)>0$. Since this property is invariant by executing the transitions in $T_0$, we deduce that $\alpha(p)>0$ and $\alpha(q)>0$. By executing $t_{\bar{p}}$ and $t_{\bar{q}}$ the right number of times, we get from $\alpha$ a configuration $\beta$ such that $\beta(\bar{i})=\beta(\bar{p})=\beta(\bar{q})=0$. This configuration is $1$-output stable. We have proved that the protocol stably computes $\phi_{i\geq n}$.
\end{example}

\medskip

In~\cite{DBLP:conf/stacs/BlondinEJ18}, it is exhibited an infinite set of natural numbers $n$ for which there exists a conservative population protocol stably computing the counting predicate $(i\geq n)$ with an interaction-width bounded by $2$, a number of leaders bounded by $2$, and a number of states bounded by $O(\log\log(n))$. This paper left open the optimality of that bound. In this paper we show that such a bound is almost optimal by proving the following main theorem.
\begin{theorem}\label{thm:main}
  For every protocol $(P,\xrightarrow{*},\rho_L,I,F)$ with a finite interaction-width that stably computes $\phi_{i\geq n}$, we have: 
$$n\leq (4+4\com{\xrightarrow{*}}+2|\rho_L|)^{|P|^{(|P|+2)^2}}$$
\end{theorem}

We deduce the following state complexity lower-bound as a corollary.
\begin{corollary}\label{cor:main}
  Let $h<\frac{1}{2}$ and let $m\geq 1$. The number of states of a protocol stably computing the $(i\geq n)$ predicate with an interaction-width bounded by $m$ and a number of leaders bounded by $m$ is at least $\Omega((\log\log(n))^h)$.
\end{corollary}
\begin{proof}
  Let us consider $\varepsilon>0$ such that $\frac{1}{2+\varepsilon}\geq h$. Notice that for $d\in\setN$ large enough, we have $d\leq 2^{(d+2)^\varepsilon}$. It follows that $d^{(d+2)^2}\leq 2^{(d+2)^{2+\varepsilon}}$ for $d$ large enough. Theorem~\ref{thm:main} shows that a protocol stably computing the $(i\geq n)$ predicate with an interaction-width bounded by $m$ and a number of leaders bounded by $m$ satisfies:
  $$n\leq (10m)^{|P|^{(|P|+2)^2}}$$
  It follows that if $n$ is large enough, then $|P|$ is large enough to satisfy $|P|^{(|P|+2)^2}\leq 2^{(|P|+2)^{2+\varepsilon}}$. It follows that $\log\log(n)\leq \log\log(10m)+\log(2)(|P|+2)^{2+\varepsilon}$. In particular:
  $$|P|\geq \left(\frac{\log\log(n)-\log\log(10m)}{\log(2)}\right)^h-2$$
  We have proved the corollary.
\end{proof}

\section{Small Stable Configurations}\label{sec:stab}
A $P$-configuration $\rho$ is said to be \emph{$(T,F)$-stabilized} where $T$ is a $P$-Petri net, and $F$ is a subset of $P$ if for every configuration $\beta$ such that $\rho\xrightarrow{T^*}\beta$, we have $\beta(p)=0$ for every $p\in P\backslash F$. In this section, we show that $(T,F)$-stabilized configurations are characterized by ``small values''. This result is obtained by applying classical Rackoff's techniques originally introduced for the Petri net coverability problem. The definition of stabilized configurations is related to the output stable configurations of protocols as shown by the following lemma.
\begin{lemma}\label{lem:stable}
  Let $(P,\xrightarrow{T^*},\rho_L,I,\gamma)$ be a protocol where $T$ is a Petri net and let $F\eqdef \gamma^{-1}(\{0\})$. A configuration is $(T,F)$-stabilized if, and only if, it is $0$-output stable. 
\end{lemma}
\begin{proof}
  By definition.
\end{proof}

\medskip

We first introduce some notations.
Given a $P$-configuration $\rho$, we introduce $\norm{\rho}_\infty\eqdef\max_{p\in P}\rho(p)$. Given a transition $t=(\alpha_t,\beta_t)$, we also introduce $\norm{t}_\infty\eqdef\max\{\norm{\alpha_t}_\infty,\norm{\beta_t}_\infty\}$. Given a Petri net $T$, we define $\norm{T}_\infty\eqdef\max_{t\in T}\norm{t}_\infty$.
Given a finite set $Q$ of states, we define several restrictions related to $Q$ as follows. Let us recall that given a $P$-configuration $\rho$, we previously defined $\rho|_Q$ as the $Q$-configuration defined by $\rho|_Q(q)\eqdef \rho(q)$ if $q\in Q$, and zero otherwise. Given a $P$-transition $t=(\alpha_t,\beta_t)$, we define the $Q$-transition $t|_Q$ as the pair $t|_Q\eqdef (\alpha_t|_Q,\beta_t|_Q)$. Given a $P$-Petri net $T$, we introduce the $Q$-Petri net $T|_Q\eqdef\{t|_Q \mid t\in T\}$. Given a word $\sigma=t_1\ldots t_k$ of $P$-transitions, we introduce the word $\sigma|_Q=t_1|_Q\ldots t_k|_Q$ of $Q$-transitions. Notice that $\alpha\xrightarrow{\sigma}\beta$ for some $P$-configurations $\alpha,\beta$ implies $\alpha|_Q\xrightarrow{\sigma|_Q}\beta|_Q$. The converse property is true in some cases as shown by the following lemma.
\begin{lemma}\label{lem:large}
  Assume that $\alpha|_Q\xrightarrow{\sigma|_Q}\rho$ for some $P$-configurations $\alpha,\rho$, some word $\sigma$ of transitions in a $P$-Petri net $T$, and some finite set $Q$. If $\alpha(p)\geq |\sigma|\norm{T}_\infty$ for every $p\in P\backslash Q$ then there exists a configuration $\beta$ such that $\alpha\xrightarrow{\sigma}\beta$, $\beta|_Q=\rho$, and $\beta(p)\geq \alpha(p)-|\sigma|\norm{T}_\infty$ for every $p\in P\backslash Q$.
\end{lemma}
\begin{proof}
  Simple induction on $|\sigma|$.
\end{proof}

\medskip

\medskip

A configuration $\rho$ is said to be \emph{$T$-coverable} from a configuration $\alpha$ where $T$ is a Petri net if there exists a word $\sigma\in T^*$ such that $\alpha\xrightarrow{\sigma}\beta\geq \rho$ for some configuration $\beta$. The minimal length of such a word $\sigma$ can be bounded using Rackoff's techniques with respect to $\norm{\rho}_\infty$, $\norm{T}_\infty$, and $|Q|$ as shown by the following result introduced in \cite{DBLP:journals/tcs/Rackoff78} to prove that the $T$-coverability problem is decidable in exponential space. 
\begin{lemma}[\cite{DBLP:journals/tcs/Rackoff78}]\label{lem:rackoff}
  If a configuration $\rho$ is $T$-coverable from a configuration $\alpha$ where $T$ is a $P$-Petri net, then there exists $\sigma\in T^*$ with a length bounded by $(\norm{\rho}_\infty+\norm{T}_\infty)^{{|P|}^{|P|}}$, and a configuration $\beta$ such that $\alpha\xrightarrow{\sigma}\beta\geq \rho$.
\end{lemma}
\begin{proof}
  This is a classical result obtained by Rackoff in~\cite{DBLP:journals/tcs/Rackoff78} by induction on $|P|$. %A complete proof with our notations is recalled in appendix.
\end{proof}

We deduce the following lemma that shows that $(T,F)$-stabilized configurations are characterized by ``small values'' (the values $\rho(p)$ for $p\in R$). In the statement of that lemma the relation $\leq$ over the $P$-configurations is defined by $\alpha\leq \beta$ if there exists a $P$-configuration $\rho$ such that $\beta=\alpha+\rho$.
\begin{lemma}\label{lem:basis}
  Let $\rho$ be a $(T,F)$-stabilized $P$-configuration with $F\subseteq P$, let $h$ be a positive integer satisfying $h\geq \norm{T}_\infty(1+\norm{T}_\infty)^{{|P|}^{|P|}}$, and let $R\eqdef \{p\in P \mid \rho(p)< h\}$. Every $P$-configuration $\alpha$ such that $\alpha|_R\leq \rho|_R$ is $(T,F)$-stabilized.
\end{lemma}
\begin{proof}
  Let us consider a $P$-configuration $\alpha$ such that $\alpha|_R\leq \rho|_R$. There exists a $R$-configuration $\mu$ such that $\rho|_R=\alpha|_R+\mu$. Assume by contradiction that $\alpha$ is not $(T,F)$-stabilized. It follows that there exists a configuration $\beta$ such that $\alpha\xrightarrow{T^*}\beta$ and $\beta(p)>0$ for some place $p\in P\backslash F$. Since $p\not\in F$ and $\rho$ is $(T,F)$-stabilized, we deduce that $\rho(p)=0$. In particular $p\in R$ since $h>0$. Since $p|_P$ is $T$-coverable from $\alpha$, Lemma~\ref{lem:rackoff} shows there exists a word $\sigma\in T^*$ of length bounded by $(1+\norm{T}_\infty)^{{|P|}^{|P|}}$ and a configuration $\eta$ such that $\alpha\xrightarrow{\sigma}\eta\geq p|_P$. It follows that $\alpha|_R\xrightarrow{\sigma|_R}\eta|_R$. From this relation and $\rho_R=\alpha|_R+\mu$, we deduce that $\rho_R\xrightarrow{\sigma|_R}\eta|_R+\mu$. Lemma~\ref{lem:large} shows that there exists a configuration $\delta$ such that $\rho\xrightarrow{\sigma}\delta$ and $\delta|_R=\eta|_R+\mu$. Since $p\in R$, we deduce that $\delta(p)=\eta(p)+\mu(p)\geq p|_P(p)=1$. As $p\not\in F$, it follows that $\rho$ is not $(T,F)$-stabilized and we get a contradiction. We have proved the lemma.
\end{proof}

\begin{remark}
  A similar result was provided in~\cite{DBLP:conf/podc/CzernerE21} in the context of conservative Petri nets with interaction-width bounded by two.
\end{remark}

\section{Bottom Configurations}\label{sec:bot}
Let $T$ be a $P$-Petri net. The \emph{$T$-component} of a $P$-configuration $\rho$ is the set of configurations $\beta$ such that $\rho\xrightarrow{T^*}\beta\xrightarrow{T^*}\rho$. A configuration $\rho$ is said to be \emph{$T$-bottom} if its $T$-component is finite and every configuration $\beta$ such that $\rho\xrightarrow{T^*}\beta$ satisfies $\beta\xrightarrow{T^*}\rho$.

\medskip

In this section we prove the following theorem that intuitively provides a way to reach with short words kind of bottom configurations with small size (small and short meaning doubly-exponential in that context). Other results proved in this section are only used for proving the following theorem and are no longer used in the sequel.
\begin{theorem}\label{thm:extract}
  Let $T$ be a $P$-Petri net, let $\rho$ be a $P$-configuration, and let $b\eqdef(4+4\norm{T}_\infty+2\norm{\rho}_\infty)^{d^d(1+(2+d^d)^{d+1})}$ with $d\eqdef|P|$. There exist two words $\sigma,w\in T^*$, a set of places $Q\subseteq P$, and two $P$-configurations $\alpha,\beta$ such that:
  \begin{itemize}
  \item $\rho\xrightarrow{\sigma}\alpha\xrightarrow{w}\beta$.
  \item $\alpha|_Q=\beta|_Q$.
  \item $\alpha(p)<\beta(p)$ for every state $p\in P\backslash Q$.
  \item $\alpha|_Q$ is $T|_Q$-bottom.
  \item The cardinal of the $T|_Q$-component of $\alpha|_Q$ is bounded by $b$.
  \item $|\sigma|,|w|,d\norm{\alpha}_\infty,d\norm{\beta}_\infty\leq b$.
  \end{itemize}
\end{theorem}

The proof of the previous theorem is obtained by iterating the following lemma in order to obtain an increasing sequence of sets $Q$.
\begin{lemma}\label{lem:extract}
  Let $T$ be a $P$-Petri net, let $\rho$ be a $P$-configuration, let $Q$ be a set of states included in $P$ such that $\rho|_Q$ is $T|_Q$-bottom, let $s$ be the cardinal of the $T|_Q$-component of $\rho|_Q$, and let $d\eqdef |P\backslash Q|$.

  There exist a word $\sigma\in T^*$ such that $|\sigma|\leq (1+d(1+s\norm{T}_\infty+\norm{\rho}_\infty)^{d^d})s$, and a $P$-configuration $\rho'$ such that $\rho\xrightarrow{\sigma}\rho'$ and such that:
  \begin{itemize}
  \item either $\rho'|_Q=\rho|_Q$ and $\rho'(p)>\rho(p)$ for every $p\in P\backslash Q$,
  \item or there exists a set $Q'\subseteq P$ that strictly contains $Q$ such that $\rho'|_{Q'}$ is $T|_{Q'}$-bottom and the cardinal $s'$ of the $T|_{Q'}$-component of $\rho'|_{Q'}$ satisfies:
    $$s'\leq (1+d(1+s\norm{T}_\infty+\norm{\rho}_\infty)^{d^d})s$$
  \end{itemize}
\end{lemma}
\begin{proof}
  Let us introduce the sequence $\lambda_1,\ldots,\lambda_d$ of natural numbers satisfying $\lambda_d\eqdef 1+s\norm{T}_\infty+\norm{\rho}_\infty$ and satisfying $\lambda_n \eqdef s\lambda_{n+1}^{d-n}\norm{T}_\infty+\lambda_{n+1}$ for every $n\in\{1,\ldots,d-1\}$. Observe that $\lambda_1\geq \cdots \geq \lambda_d$. Moreover, $\lambda_n\leq \lambda_d \lambda_{n+1}^{d-n}$ for every $1\leq n<d$. We deduce by induction that $\lambda_1^d\leq \lambda_d^{d^d}$.
  
  Let $\rho_0\eqdef \rho$. We are going to build by induction on $n$ a sequence $\rho_1,\ldots,\rho_n$ of configurations, a sequence $\sigma_1,\ldots,\sigma_n$ of words in $T^*$, and a sequence $p_1,\ldots,p_n$ of states in $P\backslash Q$ such that for every $i\in\{1,\ldots,n\}$ we have:
  \begin{itemize}
  \item[(i)] $\rho_{i-1}\xrightarrow{\sigma_i}\rho_i$.
  \item[(ii)] $|\sigma_i|\leq \lambda_i^{d-i+1} s$.
  \item[(iii)] $\rho_i(p)\geq \lambda_i$ for every $p\in\{p_1,\ldots,p_i\}$.
  \end{itemize}
  So, let us assume that $\rho_1,\ldots,\rho_n$, $\sigma_1,\ldots,\sigma_n$, and $p_1,\ldots,p_n$ are built for some $n\geq 0$. Since $p_1,\ldots,p_n$ are distinct elements in $P\backslash Q$, it follows that $n\leq d$.

  Let us first assume that $n=d$. In that case, we have $\rho_d(p)\geq \lambda_d$ for every $p\in P\backslash Q$. As $\rho|_Q$ is a $T|_Q$-bottom configuration and $\rho|_Q\xrightarrow{(\sigma_1\ldots\sigma_d)|_Q}\rho_d|_Q$ and since the cardinal of the $T|_Q$-component of $\rho|_Q$ is bounded by $s$, we deduce that there exists a word $w\in T^*$ such that $\rho_d|_Q\xrightarrow{w|_Q}\rho|_Q$ and $|w|< s$. Since $\lambda_d\geq s\norm{T}_\infty\geq |w|\norm{T}_\infty$, Lemma~\ref{lem:large} shows that $\rho_d\xrightarrow{w}\rho'$ for some configuration $\rho'$ such that $\rho'|_Q=\rho|_Q$ and such that for every $p\in P\backslash Q$ we have $\rho'(p)\geq \rho(p)-|w|\norm{T}_\infty\geq \lambda_d-s\norm{T}_\infty>\rho(p)$ by definition of $\lambda_d$. Let us introduce $\sigma\eqdef \sigma_1\ldots\sigma_d w$ and notice that $|\sigma|\leq (\lambda_1^d+\cdots+\lambda_d^{d-d+1}+1)s\leq (1+d\lambda_1^d)s$ and we have proved that the lemma holds (first case).

  So we can assume that $n<d$. Let us introduce the set $R_n=(P\backslash Q)\backslash \{p_1,\ldots,p_n\}$. Since $|R_n|=d-n$, the set $R_n$ is non empty.

  Assume first that for every configuration $\beta$ such that $\rho_n\xrightarrow{T^*}\beta$ we have $\beta(p)<\lambda_{n+1}$ for every $p\in R_n$. In that case let $Q'\eqdef Q\cup R_n$. It follows that the cardinal of the set of configurations $\beta'$ such that $\rho_n|_{Q'}\xrightarrow{T|_{Q'}^*}\beta'$ is bounded by $s\lambda_{n+1}^{d-n}$. Hence, there exists a configuration $\beta'$ that is $T|_{Q'}$-bottom and a word $w\in T^*$ such that $\rho_n|_{Q'}\xrightarrow{w|_{Q'}}\beta'$ and such that $|w|< s\lambda_{n+1}^{d-n}$. Notice that the cardinal $s'$ of the $T|_{Q'}$-component of $\beta'$ is bounded by $s\lambda_{n+1}^{d-n}\leq s\lambda_1^d$. As $\rho_n(p)\geq \lambda_n$ for every $p\in\{p_1,\ldots,p_n\}$ and $\lambda_n\geq (s\lambda_{n+1}^{d-n}-1)\norm{T}_\infty\geq |w|\norm{T}_\infty$, Lemma~\ref{lem:large} shows that there exists a configuration $\rho'$ such that $\rho_n\xrightarrow{w}\rho'$, and  $\rho'|_{Q'}=\beta'$. Let us consider the word $\sigma\eqdef \sigma_1\ldots\sigma_n w$. Notice that $|\sigma|\leq (\lambda_1^d+\cdots+\lambda_n^{d-n+1}+\lambda_{n+1}^{d-n}) s\leq d\lambda_1^d s$ and we have proved that the lemma holds (second case).

  Finally, assume that there exists a configuration $\rho_{n+1}$ such that $\rho_n\xrightarrow{\sigma_{n+1}}\rho_{n+1}$ for some word $\sigma_{n+1}\in T^*$ and such that $\rho_{n+1}(p_{n+1})\geq \lambda_{n+1}$ for some state $p_{n+1}\in R_n$. We assume that $|\sigma_{n+1}|$ is minimal. Observe that every intermediate configuration $\beta$ such that $\rho_n\xrightarrow{u}\beta\xrightarrow{v}\rho_{n+1}$ with $u v=\sigma_{n+1}$ and $|v|\geq 1$ satisfies $\beta(p)<\lambda_{n+1}$ for every $p\in R_n$ by minimality of $|\sigma_{n+1}|$. We deduce that there exists a word $w\in T^*$ such that $\rho_n|_{Q\cup R_n}\xrightarrow{w|_{Q\cup R_n}}\rho_{n+1}|_{Q\cup R_n}$ and such that $|w|\leq s\lambda_{n+1}^{d-n}$. As $\rho_n(p)\geq \lambda_n$ for every $p\in \{p_1,\ldots,p_n\}$ and $\lambda_n\geq s\lambda_{n+1}^{d-n}\norm{T}_\infty\geq |w|\norm{T}_\infty$, Lemma~\ref{lem:large} shows that there exists a configuration $\beta$ such that $\rho_n\xrightarrow{w}\beta$ and $\beta|_{Q\cup R_n}=\rho_{n+1}|_{Q\cup R_n}$. In particular $\beta(p_{n+1})=\rho_{n+1}(p_{n+1})\geq \lambda_{n+1}$. By minimality of $|\sigma_{n+1}|$, we get $|\sigma_{n+1}|\leq |w|\leq s\lambda_{n+1}^{d-n}$. Now, observe that for every $p\in\{p_1,\ldots,p_n\}$ we have $\rho_{n+1}(p)\geq \rho_n(p)-|\sigma_{n+1}|\norm{T}_\infty\geq \lambda_n-s\lambda_{n+1}^{d-n}\norm{T}_\infty\geq \lambda_{n+1}$ by definition of $\lambda_n$. We have extended our sequence in such a way $(i)$, $(ii)$, and $(iii)$ are fulfilled.
  
We have proved the lemma.
\end{proof}

Now, let us prove Theorem~\ref{thm:extract}. Observe that if $d=0$ the theorem is trivial. So, we can assume that $d\geq 1$.
%\begin{proof}
  Let $Q_0\eqdef \emptyset$, $\rho_0\eqdef\rho$, and $s_0\eqdef 1$. Notice that $\rho_0|_{Q_0}$ is $T|_{Q_0}$-bottom and the cardinal of the $T|_{Q_0}$-component of $\rho_0|_{Q_0}$ contains $s_0$ elements. We build by induction on $n$ a sequence $Q_1,\ldots,Q_n$ of subsets of $P$, a sequence $\rho_1,\ldots,\rho_n$ of configurations, a sequence $\sigma_1,\ldots,\sigma_n$ of words in $T^*$ such that for every $i\in\{1,\ldots,n\}$:
  \begin{itemize}
  \item $\rho_{i-1}\xrightarrow{\sigma_i}\rho_i$.
  \item $\rho_i|_{Q_i}$ is $T|_{Q_i}$-bottom.
  \item The cardinal of the $T|_{Q_i}$-component of $\rho_i|_{Q_i}$ is equal to $s_i$.
  \item $Q_{i-1}\subset Q_i$.
  \item $|\sigma_i|,s_i\leq (1+d(1+s_{i-1}\norm{T}_\infty+\norm{\rho_{i-1}}_\infty)^{d^{d}})s_{i-1}$.
  \end{itemize}
  Assume the sequence built for some $n\geq 0$. Lemma~\ref{lem:extract} on the configuration $\rho_n$ and the set $Q_n$ shows that there exist a word $\sigma_{n+1}$ such that $|\sigma_{n+1}|\leq (1+d(1+s_{n}\norm{T}_\infty+\norm{\rho_{n}}_\infty)^{d^d})s_n$, and a configuration $\rho_{n+1}$ such that $\rho_n\xrightarrow{\sigma_n}\rho_{n+1}$, such that:
  \begin{itemize}
  \item either $\rho_{n+1}|_{Q_{n}}=\rho_n|_{Q_n}$ and $\rho_{n+1}(p)>\rho_n(p)$ for every $p\in P\backslash Q_n$,
  \item or there exists $Q_{n+1}$ such that $Q_n\subset Q_{n+1}\subseteq P$ such that $\rho_{n+1}|_{Q_{n+1}}$ is $T|_{Q_{n+1}}$-bottom and the cardinal $s_{n+1}$ of its $T|_{Q_{n+1}}$-component satisfies:
    $$s_{n+1}\leq (1+d(1+s_n\norm{T}_\infty+\norm{\rho_n}_\infty)^{d^d})s_n$$
  \end{itemize}
  Observe that in the second case we have extended the sequences. In the first case, let $\alpha\eqdef \rho_n$, $\beta \eqdef \rho_{n+1}$, $\sigma\eqdef \sigma_1\ldots\sigma_n$, $w\eqdef \sigma_{n+1}$, and $Q \eqdef Q_{n}$. Since $Q_0\subset Q_1\cdots \subset Q_n$ are subsets of $P$, we deduce that $n\leq d$. Let us introduce $a=(1+d)(2+2\norm{T}_\infty+\norm{\rho}_\infty)^{d^d}$ and $h=2+d^d$ and let us prove by induction on $i$ that we have $|\sigma_i|,|s_i|\leq a^{h^i}$ and $\norm{\rho_i}_\infty\leq (1+\norm{T}_\infty)a^{h^i}$
  with the convention $\sigma_0=\varepsilon$.

  The rank $i=0$ is immediate. Assume the rank $i-1$ proved. We have:
  \begin{align*}
    |\sigma_i|,s_i
    & \leq (1+d(1+s_{i-1}\norm{T}_\infty+\norm{\rho_{i-1}}_\infty)^{d^{d}})s_{i-1}\\
    & \leq (1+d)(2+2\norm{T}_\infty)^{d^d}a^{h^{i-1}(d^d+1)}\\
    & \leq a^{1+h^{i-1}(h-1)}\\
    &\leq a^{h^i}
  \end{align*}
  Since $\rho_{i-1}\xrightarrow{\sigma_i}\rho_i$, we deduce that $\norm{\rho_i}_\infty\leq \norm{\rho_{i-1}}_\infty+|\sigma_i|\norm{T}_\infty\leq (1+\norm{T}_\infty)a^{h^i}$. The induction is proved.

  It follows that $|\sigma|\leq d a^{h^d}$, $|w|\leq a^{h^{d+1}}$, and $d\norm{\alpha}_\infty,d\norm{\beta}_\infty\leq d(1+\norm{T}_\infty)a^{h^{d+1}}\leq a^{1+h^{d+1}}$. Since $d\geq 1$, we deduce that $(1+d)\leq 2^{d^d}$. In particular $a\leq (4+4\norm{T}_\infty+2\norm{\rho}_\infty)^{d^d}$. We have proved Theorem~\ref{thm:extract}.
%\end{proof}

\section{Petri Nets with Control-States}\label{sec:PNS}
A \emph{$P$-Petri net with control-states} (or simply a Petri net with control-states when the finite set of states $P$ is clear from the context) is a triple $(S,T,E)$ where $S$ is a non empty finite set of elements called \emph{control-states}, $T$ is a $P$-Petri net, and $E\subseteq S\times T\times S$ is a set of elements called \emph{edges}. The \emph{Parikh image} of a word $\pi=e_1\ldots e_k$ of edges is the mapping $\#\pi\in \setN^E$ defined by $\#\pi(e)=|\{j\in\{1,\ldots,k\} \mid e_j=e\}|$. The \emph{displacement} of a transition $t=(\alpha_t,\beta_t)$ is the function $\Delta(t)\in\setZ^P$ defined by $\Delta(t)(p)=\beta_t(p)-\alpha_t(p)$ for every $p\in P$. The \emph{displacement} of an edge $e=(s,t,s')$ is defined as $\Delta(e)\eqdef\Delta(t)$. The displacement of a word $\pi=e_1\ldots e_k$ of edges is $\Delta(\pi)\eqdef\sum_{1\leq j\leq k}\Delta(e_j)$. We denote by $|\pi|\eqdef k$ the \emph{length} of $\pi$. A \emph{path} $\pi$ from a control-state $s$ to a control-state $s'$ is a word $\pi=e_1\ldots e_k$ of edges in $E$ such that there exists control-states $s_0,\ldots,s_k$ in $S$ and transitions $t_1,\ldots,t_k$ in $T$ such that $s_0=s$, $s_k=s'$, and such that $e_j=(s_{j-1},t_j,s_j)$ for every $1\leq j\leq k$. Such a path is called a \emph{cycle} if $s=s'$. A cycle $\theta$ of a Petri net with control-states is said to be \emph{total} if $\#\theta(e)>0$ for every $e\in E$. The cycle is said to be \emph{simple} if the control-states $s_1,\ldots,s_k$ are distinct. A \emph{multicycle} $\Theta$ is a sequence $\theta_1,\ldots,\theta_k$ of cycles.  We denote by $|\Theta|\eqdef\sum_{j=1}^k|\theta_j|$, the \emph{length} of a multicycle $\Theta$. We introduce the \emph{Parikh image} $\#\Theta\eqdef \sum_{j=1}^k\#\theta_j$ and the \emph{displacement} $\Delta(\Theta)\eqdef\sum_{j=1}^k\Delta(\theta_j)$ of such a multicycle $\Theta$. A multicycle $\Theta$ is said to be \emph{total} if $\#\Theta(e)>0$ for every $e\in E$.

\medskip

A \emph{Petri net with control-states} $(S,T,E)$ is said to be \emph{strongly connected} if for every pair $(s,s')$ of control-states in $S$, there exists a path from $s$ to $s'$.  Let us recall the classical Euler lemma in the context of Petri nets with control-states.
\begin{lemma}[Euler Lemma]\label{lem:euler}
  For every total multicycle $\Theta$ in a strongly connected Petri net with control-states there exists a total cycle $\theta$ such that $\#\theta=\#\Theta$.
\end{lemma}

We deduce the following lemma.
\begin{lemma}\label{lem:total}
For any strongly connected Petri net with control-states $(S,T,E)$, there exists a total cycle $\theta$ with a length bounded by $|E||S|$.
\end{lemma}
\begin{proof}
  Every edge $e\in E$ occurs in at least one simple cycle $\theta_e$. It follows that the multicycle $\Theta=(\theta_e)_{e\in E}$ is total. From Lemma~\ref{lem:euler} we deduce that there exists a total cycle $\theta$ such that $\#\theta=\#\Theta$. Notice that $|\theta|=\sum_{e\in E}|\theta_e|\leq |E||S|$.
\end{proof}

A mapping $a\in\setZ^P$ is called a \emph{$P$-action} (or simply an action if $P$ is clear from the context). Notice that displacements of transitions, edges, paths, and multicyles are actions. We associate with an action $a$ the value $\norm{a}_1\eqdef\sum_{p\in P}|a(p)|$. Given a finite set $Q$, we denote by $a|_Q$ the action defined for every $q\in Q$ by $a|_Q(q)\eqdef a(q)$ if $q\in P$, and zero otherwise. 

\begin{lemma}\label{lem:multi}
  Let $\Theta$ be a multicycle of a $P$-Petri net with control-states $(S,T,E)$ with $\norm{T}_\infty>0$, let $Q\subseteq P$, let $d\eqdef |P|$, and let $k> \norm{\Delta(\Theta)|_Q}_1(1+2|S|\norm{T}_\infty)^{d(d+1)}$.

  There exists a multicycle $\Theta'$ such that:
  \begin{itemize}
  \item For every $p\in P$ we have:
    \begin{itemize}
    \item $\Delta(\Theta')(p)\leq 0$ if $\Delta(\Theta)(p)\leq 0$.
    \item $\Delta(\Theta')(p)<0$ if $\Delta(\Theta)(p)\leq -k $.
    \item $\Delta(\Theta')(p)\geq 0$ if $\Delta(\Theta)(p)\geq 0$.
    \item $\Delta(\Theta')(p)>0$ if $\Delta(\Theta)(p)\geq k$.
    \end{itemize}
  \item For every $q\in Q$ we have $\Delta(\Theta')(q)=0$.
  \item For every edge $e\in E$ we have $\#\Theta'(e)>0$ if $\#\Theta(e)\geq k$.
  \item $|\Theta'|\leq (|E|+d)(1+2|S|\norm{T}_\infty)^{d(d+1)}$.
  \end{itemize}
\end{lemma}
\begin{proof}
  Since every cycle can be decomposed into a sequence of simple cycles without changing the Parikh image, we can assume without loss of generality that $\Theta$ is a sequence of simple cycles. We introduce the set $A$ of actions $\Delta(\theta)$ where $\theta$ ranges over the simple cycles, and $n\eqdef|A|$ its cardinal. Notice that for every $a\in A$ and for every $p\in P$, we have $|a(p)|\leq |S|\norm{T}_\infty$. It follows that $n\leq (1+2|S|\norm{T}_\infty)^d$.
  
  We denote by $s$ the sign function of $a$ formally defined by $s(\vr{c})\eqdef 1$ if $\Delta(\Theta)(\vr{c})\geq 0$, $s(\vr{c})\eqdef -1$ otherwise. We also introduce the $P$-configuration $f$ defined by $f(p)\eqdef |\Delta(\Theta)(p)|$ for every $p\in P$, and the function $g:A\rightarrow\setN$ such that $g(a)$ is the number of simple cycle $\theta$ that occurs in $\Theta$ such that $\Delta(\theta)=a$.
  
  Notice that $s(p)f(p)=\sum_{a\in A}g(a)a(p)$ for every $p\in P$. We introduce the following linear system over the free variables $(\alpha,\beta)\in\setN^P\times\setN^A$:
  \begin{equation}\label{eq:system}
    \bigwedge_{p\in P}s(p)\alpha(p)=\sum_{a\in A}\beta(a)a(p)
  \end{equation}
  Notice that $(f,g)$ is a solution of that system. From~\cite{Pottier:1991:MSL:647192.720494}, there exists a finite set $H$ of solutions $(\alpha,\beta)$ of that system such that $(f,g)=\sum_{(\alpha,\beta)\in H}(\alpha,\beta)$ and such that for every $(\alpha,\beta)\in H$, we have $\norm{\alpha}_1+\norm{\beta}_1\leq (2+\sum_{a\in A}\norm{a}_\infty)^d$. As $\sum_{a\in A}\norm{a}_\infty\leq (1+2|S|\norm{T}_\infty)^d|S|\norm{T}_\infty$ we deduce (by using $|S|\norm{T}_\infty\geq 1$):
  \begin{equation}\label{eq:bound}
    \norm{\alpha}_1+\norm{\beta}_1\leq (1+2|S|\norm{T}_\infty)^{d(d+1)}
  \end{equation}
  We introduce the set $H_0$ of pairs $(\alpha,\beta)\in H$ such that $\alpha(q)=0$ for every $q\in Q$. Observe that $\sum_{q\in Q}\sum_{(\alpha,\beta)\in H}\alpha(q)$ is equal to $\sum_{q\in Q}|\Delta(\Theta)(q)|=\norm{\Delta(\Delta)|_Q}_1$, and it is also equals to $\sum_{q\in Q}\sum_{(\alpha,\beta)\in H\backslash H_0}\alpha(q)\geq |H\backslash H_0|$. In particular we have:
  $$|H\backslash H_0|\leq \norm{\Delta(\Delta)|_Q}_1$$

  We introduce the set $F$ of edges $e\in E$ such that $\Theta(e)\geq k$. Let $e\in F$. The sum $\sum_{(\alpha,\beta)\in H}\beta(e)$ is equals to $\#\Theta(e)$ and it is also equals to $\sum_{(\alpha,\beta)\in H\backslash H_0}\beta(e)+\sum_{(\alpha,\beta)\in H_0}\beta(e)$. As $\sum_{(\alpha,\beta)\in H\backslash H_0}\beta(e)\leq |H\backslash H_0|(1+2|S|\norm{T}_\infty)^{d(d+1)}$ we deduce that $\sum_{(\alpha,\beta)\in H_0}\beta(e)>0$. In particular there exists $(\alpha,\beta)\in H_0$ such that $\beta(e)>0$.

  We also introduce the set $R$ of $p\in P$ such that $|\Delta(\Theta)(p)|\geq k$. Let $p\in R$. The sum $\sum_{(\alpha,\beta)\in H}\alpha(p)$ is equals to $|\Delta(\Theta)(p)|$ and it is also equals to $\sum_{(\alpha,\beta)\in H\backslash H_0}\alpha(p)+\sum_{(\alpha,\beta)\in H_0}\alpha(p)$. As $\sum_{(\alpha,\beta)\in H\backslash H_0}\alpha(p)\leq |H\backslash H_0|(1+2|S|\norm{T}_\infty)^{d(d+1)}$ we deduce that $\sum_{(\alpha,\beta)\in H_0}\alpha(p)>0$. In particular there exists $(\alpha,\beta)\in H_0$ such that $\alpha(p)>0$.

  Now, let us introduce for each $e\in F$ a pair $(\alpha_e,\beta_e)\in H_0$ such that $\beta_e(e)>0$, and let us introduce for each $p\in R$ a pair $(\alpha_p,\beta_p)\in H_0$ such that $\alpha_p(p)>0$. Let us introduce  $(\alpha',\beta')\eqdef \sum_{e\in F}(\alpha_e,\beta_e)+\sum_{p\in R}(\alpha_p,\beta_p)$ and observe that $\alpha'(e)>0$ for every $e\in F$, $\beta'(p)>0$ for every $p\in R$, and $\beta'(q)=0$ for every $q\in Q$. Moreover, since $(\alpha',\beta')$ is a solution of (\ref{eq:system}), it follows that there exists a multicycle $\Theta'$ such that $\#\Theta'=\beta'$. In particular $\Delta(\Theta')=\Delta(\beta')$. Notice that $\Delta(\Theta')=\alpha'$, and $|\Theta'|=\norm{\beta'}_1\leq (|F|+|R|)(1+2|S|\norm{T}_\infty)^{d(d+1)}\leq (|E|+d)(1+2|S|\norm{T}_\infty)^{d(d+1)}$ and we have proved the lemma.
\end{proof}

\section{Proof of Theorem~\ref{thm:main}}\label{sec:main}
In this section we provide a proof of Theorem~\ref{thm:main}.

%% The following simple technical lemma will be useful two times in this proof.
%% \begin{lemma}\label{lem:nz}
%%   Let $T$ be a Petri net, let $\rho$ be a configuration from which the zero configuration $\mathbb{0}$ is not $T$-reachable, and $\rho'$ be a configuration $T$-reachable from $\rho$ that is $(T,F')$-stabilized for some set $F'$. Then $\stab_T(\rho)$ has a non empty intersection with $F'$.
%% \end{lemma}
%% \begin{proof}
%%   Let $Q\eqdef \stab_T(\rho)$. Since $\rho$ is $(T,Q)$-stabilizable, we deduce that $\rho'$ is also $(T,Q)$-stabilizable. Lemma~\ref{lem:inter} shows that $\rho'$ is $(T,Q\cap F')$-stabilizable. It follows that that exists a configuration $\beta$ that is $T$-reachable from $\rho'$ such that $\used{\beta}\subseteq Q\cap F'$. Since $\beta$ is reachable from $\rho$, we deduce that $\beta\not=\mathbb{0}$. Hence $\used{\beta}$ is not empty. It follows that $Q\cap F'$ is non empty.
%% \end{proof}

\medskip

\newcommand{\Px}{{P'}}

We consider a finite interaction-width protocols $(P,\xrightarrow{*},\rho_L,I,\gamma)$ that stably computes the $(n>i)$ predicate. Notice that $I=\{i\}$. We introduce the Petri net $T\eqdef\{t\in \xrightarrow{*} \mid \com{t}\leq \com{\xrightarrow{*}}\}$. Let us recall that the additive preorder $\xrightarrow{*}$ is equal to $\xrightarrow{T^*}$. Let $d\eqdef|Q|$ and $F=\gamma^{-1}(\{0\})$. Notice that if $d=1$ then $n=1$ and the proof of the theorem is done in that case. So, we can assume that $d\geq 2$. We introduce the following numbers:
\begin{align*}
  b&\eqdef(4+4\norm{T}_\infty+2\norm{\rho_L}_\infty)^{(d-1)^{d-1}(1+(2+(d-1)^{d-1})^d)}\\
  h&\eqdef d(1+\norm{T}_\infty)b\\
  k&\eqdef d h^{d^2+d+1}\\
  a&\eqdef h^{2d+3}\\
  \ell&\eqdef h^{5d^2}\\
  r&\eqdef 2(d-1)^{d-1}(1+(2+(d-1)^{d-1})^d)(5d^2+2d+4)
\end{align*}

\medskip

We introduce $\Px\eqdef P\backslash I$. It follows that $|\Px|= d-1$. Theorem~\ref{thm:extract} applied on the Petri net $T|_\Px$ and the configuration $\rho_L|_\Px$ shows that there exist two words $\sigma,w\in T^*$, a set $Q\subseteq \Px$, and two configurations $\alpha,\beta$ such that:
\begin{itemize}
\item $\rho_L|_\Px\xrightarrow{\sigma|_\Px}\alpha\xrightarrow{w|_\Px}\beta$.
\item $\alpha|_Q=\beta|_Q$.
\item $\alpha(p)<\beta(p)$ for every $p\in \Px\backslash Q$
\item $\alpha|_Q$ is $T|_Q$-bottom.
\item The cardinal of the $T|_Q$-component of $\alpha|_Q$ is bounded by $b$.
\item $|\sigma|,|w|,d\norm{\alpha}_\infty,d\norm{\beta}_\infty\leq b$.
\end{itemize}
Notice that $|T|\leq (1+2\norm{T}_\infty)^{2d}\leq h^{2d}$.

\medskip

We introduce the Petri net with control-states $(S,T,E)$ where $S$ is the $T|_Q$-component of $\alpha|_Q$, and $E$ is the set of edges $(s,t,s')\in S\times T\times S$ such that $s\xrightarrow{t|_Q}s'$. Observe that $|E|\leq |S||T|$ since for every $(s,t,s')$ in $E$ the value of $s'$ is determined by the pair $(s,t)$. It follows that we have:
$$|E|\leq h^{2d+1}$$

Lemma~\ref{lem:total} shows that there exists a total cycle $\theta_E$ of $(S,T,E)$ with a length bounded by $|S||E|$. Without loss of generality we can assume that this total cycle is on the control-state $\alpha|_Q$ by considering a rotation of that cycle. We denote by $\sigma_E$ the label in $T^*$ of this total cycle. Observe that $\norm{T}_\infty|\sigma_E|\leq a$.

Since $\alpha\xrightarrow{w|_\Px}\beta$, $\alpha|_Q=\beta|_Q$, and $\alpha(p)<\beta(p)$ for every $p\in \Px\backslash Q$, we deduce that there exists a configuration $\eta$ such that $\eta(p)\geq a\ell$ for every $p\in \Px\backslash Q$, such that $\alpha|_Q=\eta|_Q$, and such that:
$$\alpha\xrightarrow{w^{a \ell}|_\Px}\eta$$

Moreover, since $\sigma_E$ is the label of a cycle on $\alpha|_Q$ we deduce that $\alpha|_Q\xrightarrow{\sigma_E|_{Q}}\alpha|_Q$. From  $\eta|_Q=\alpha|_Q$ it follows that $\eta|_Q\xrightarrow{\sigma_E^\ell|_{Q}}\alpha|_Q$. As $\eta(p)\geq a\ell\geq \norm{T}_\infty|\sigma_E^\ell|$ for every $p\in \Px\backslash Q$, Lemma~\ref{lem:large} shows that there exists a $P'$-configuration $\delta$ such that $\delta|_Q=\alpha|_Q$ and such that:
$$\eta\xrightarrow{\sigma_E^\ell|_\Px}\delta$$
Observe that $|\sigma w^{a \ell}\sigma_E^\ell|\norm{T}_\infty\leq (b+b a \ell)\norm{T}_\infty+a\ell\leq 2b a\ell(\norm{T}_\infty+1)\leq a\ell h\leq h^{2d+4}\ell$.

\medskip

Assume by contradiction that $n>h^{2d+4}\ell$, and let us introduce the configuration $\rho'$ defined by $\rho'\eqdef \rho_L+(n-1).i$ where $i$ is the state such that $I=\{i\}$. Lemma~\ref{lem:large} shows that there exist $P$-configurations $\alpha',\eta',\delta'$ such that $\alpha'|_\Px=\alpha$, $\eta'|_\Px=\eta$, $\delta'|_\Px=\delta$ and such that:
$$\rho'\xrightarrow{\sigma}\alpha'\xrightarrow{w^{\alpha \ell}}\eta'\xrightarrow{\sigma_E^\ell}\delta'$$
Since the population protocol is stably computing the $(i\geq n)$ predicate and $n-1<n$, there exists a $0$-output stable configuration $\mu$ and a word $\sigma'\in T^*$ such that $\delta'\xrightarrow{\sigma'}\mu$. Lemma~\ref{lem:stable} shows that $\mu$ is $(T,F)$-stabilized.
Observe that $w^{\alpha\ell}\sigma_E^\ell\sigma'$ is the label of a path of $(S,T,E)$ from $\alpha|_Q$ to $\mu|_Q$. It follows that the Parikh image of that path can be decomposed as the Parikh image of a multicycle $\Theta$ and the Parikh image of an elementary path $\pi$. Observe $\Delta(\Theta)+\Delta(\pi)=\Delta(w^{\alpha\ell}\sigma_E^\ell\sigma')=\mu-\alpha'$. Notice that $\#\Theta(e)\geq \ell$ for every $e\in E$ since $\sigma_E$ is the label of a total cycle on $\alpha|_Q$. Since $\pi$ is an elementary path, we deduce that $\norm{\Delta(\pi)}_1\leq d|S|\norm{T}_\infty\leq db\norm{T}_\infty\leq h-db$.

\medskip

We introduce the set $R\eqdef\{p\in P \mid \mu(p)< h\}$. Since $h\geq \norm{T}_\infty(1+\norm{T}_\infty)^{d^d}$, Lemma~\ref{lem:basis} shows that every configuration $\mu'$ such that $\mu|_R=\mu'|_R$ is $(T,F)$-stabilized. Observe that if $i\not\in R$ then $\mu+i$ is $(T,F)$-stabilized, and by additivity, we deduce that $\rho_L+n.i\xrightarrow{T^*}\mu+i$. Since $\mu+i$ is $(T,F)$-stabilized, this configuration cannot reach a $1$-output stable configuration. In particular the protocol is not stably computing the $(i\geq n)$ predicate and we get a contradiction. It follows that $i\in R$.

\medskip

We introduce $R'\eqdef R\backslash I$. Since $d\norm{\alpha}_\infty\leq b$ and $\alpha'|_\Px=\alpha|_\Px$, we deduce that $d\norm{\alpha'|_{R'}}_\infty\leq b$. From $\Delta(\Theta)=\mu-\alpha'-\Delta(\pi)$ we deduce:
$$\norm{\Delta(\Theta)|_{R'}}_1
\leq (d-1)h + b+  h-db
\leq d h$$
As $1+2|S|\norm{T}_\infty\leq 1+h-2b< h$, we deduce that $k> \norm{\Delta(\Theta)|_{R'}}_1(1+2|S|\norm{T}_\infty)^{d(d+1)}$, Lemma~\ref{lem:multi} shows that there exists a multicycle $\Theta'$ such that:
\begin{itemize}
\item For every $p\in P$ we have:
  \begin{itemize}
  \item $\Delta(\Theta')(p)\leq 0$ if $\Delta(\Theta)(p)\leq 0$.
  \item $\Delta(\Theta')(p)<0$ if $\Delta(\Theta)(p)\leq -k $.
  \item $\Delta(\Theta')(p)\geq 0$ if $\Delta(\Theta)(p)\geq 0$.
  \item $\Delta(\Theta')(p)>0$ if $\Delta(\Theta)(p)\geq k$.
  \end{itemize}
\item For every $p\in R'$ we have $\Delta(\Theta')(p)=0$.
\item For every $e\in E$ we have $\#\Theta'(e)>0$ if $\#\Theta(e)\geq k$.
\item $\norm{\Theta'}_1\leq (|E|+d)(1+2|S|\norm{T}_\infty)^{d(d+1)}$
\end{itemize}

Let $m\eqdef-\Delta(\Theta')(i)$ and let us prove that $m>0$. We have $\Delta(\Theta)(i)=\mu(i)-\alpha'(i)-\Delta(\pi)(i)$. Since $i\in R$, we get $\mu(i)< h$. Since $\rho_L+(n-1).i\xrightarrow{\sigma}\alpha'$, we deduce that $\alpha'(i)=\rho_L(i)+(n-1)-\Delta(\sigma)(i)\geq n - h$ since $|\sigma|\leq b$. We deduce that $\Delta(\Theta)(i)< h-n+h+h\leq 3h-n\leq -k$. Hence $\Delta(\Theta')(i)<0$. It follows that $m>0$.

Let $\eta\eqdef m.i+\Delta(\Theta')$. Notice that $\eta(i)=0$ and $\eta(p)=0$ for every $p\in R'$. In particular $\eta(p)=0$ for every $p\in R$. Let us prove that $\eta$ is a configuration. For every  $p \in P\backslash R$ we have $\Delta(\Theta)(p)=\mu(p)-\eta'(p)-\Delta(\pi)(p)\geq db\norm{T}_\infty+b-b-db\norm{T}_\infty\geq 0$. It follows that $\Delta(\Theta')(p)\geq 0$. In particular $\eta(p)\geq 0$. We have proved that $\eta$ is a configuration.

Finally, observe that $\#\Theta(e)\geq \ell\geq k$ for every $e\in E$. In particular $\#\Theta'(e)>0$. Lemma~\ref{lem:euler} shows that $\#\Theta'$ is the Parikh image of a cycle $\theta$ on $x|_Q$. Let $u$ be the label of that cycle. Since $|u|=\norm{\Theta'}_1$, we deduce that:
\begin{align*}
  |u|\norm{T}_\infty
  &\leq \norm{T}_\infty(|E|+d)(1+2b\norm{T}_\infty)^{d(d+1)}\\
  &\leq d(1+\norm{T}_\infty)^{2d}(1+2b\norm{T}_\infty)^{d^2+d+1}\\
  &\leq \ell
\end{align*}

Lemma~\ref{lem:large} shows that:
$$\eta'+m.i\xrightarrow{u}\eta'+\eta$$
We have proved:
$$\rho_L+(n-1+m).i\xrightarrow{\sigma w^{a \ell}u\sigma_E^\ell\sigma'}\mu+\eta$$
Since $(\mu+\eta)|_R=\mu|_R$ we deduce that $\mu+\eta$ is $(T,F)$-stabilized. It follows that this configuration cannot reach a $1$-output stable configuration. In particular the protocol is not stably computing the $(i\geq n)$ predicate and we get a contradiction. It follows that $n\leq h^{2d+4}\ell=h^{5d^2+2d+4}$.

\medskip

Notice that $d(1+\norm{T}_\infty)\leq 2^d(1+\norm{T}_\infty)^d\leq b$. Thus $h\leq b^2$. We deduce that $n\leq (4+4\norm{T}_\infty+2\norm{\rho_L}_\infty)^r$. Since $d\geq 2$, we deduce that $d^d=((d-1)+1)^d\geq (d-1)^d+d(d-1)^{d-1}+1\geq (d-1)^{d-1}+2+1$. Hence $1+(2+ (d-1)^{d-1})^d\leq 1+(d^d-1)^d\leq d^{d^2}$. Moreover, $2(d-1)^{d-1}\leq d^d$. Notice that $2d\leq d^2$ and $4\leq d^2$. Hence $5d^2+2d+4\leq 7d^2\leq d^5$ since $7\leq d^3$. We deduce that $r$ is bounded by $d^{d^2+d+3}$. As $d^2+d+3\leq (d+2)^2$, we get $r\leq d^{(d+2)^2}$.
\medskip

We have proved Theorem~\ref{thm:main} just by observing that $\norm{\rho_L}_\infty\leq |\rho_L|$ and $\norm{T}_\infty\leq \com{\xrightarrow{*}}$.

\section{Conclusion}\label{sec:conc}
This paper introduces protocols that allow agents destructions/creations and leaders. Our definition of stably computing is a straight-forward extension of the one introduced by Dana Angluin, James Aspnes, and David Eisenstat in~\cite{DBLP:conf/podc/AngluinAE06}.

\medskip

We provided in this paper state complexity lower-bounds of the form $\Omega(\log\log(n)^h)$ for any $h<\frac{1}{2}$ for protocols stably computing the counting predicates when the number of leaders and the interaction-width are bounded. This lower-bound almost matches the upper-bound $O(\log\log(n))$ introduced in~\cite{DBLP:conf/stacs/BlondinEJ18} by Blondin, Esparza, and Jaax. We left as open the exact asymptotic state complexity.

\medskip

Notice that for leaderless protocols, the state complexity is still open since there is an exponential gap between the upper-bound $O(\log(n))$ given in~\cite{DBLP:conf/stacs/BlondinEJ18} and the lower-bound introduced in this paper.

% ----------------------------------------------------------------
% you can include the acknowledgments in the source, but `anonymous' option will hide them
\begin{acks}
The author is supported by the grant ANR-17-CE40-0028 of the French National Research Agency ANR (project BRAVAS)
\end{acks}
% ----------------------------------------------------------------
% use ACM-Reference-Format for the references
\bibliographystyle{ACM-Reference-Format}
\bibliography{main}

\end{document}